\documentclass[a4paper]{scrartcl}
\usepackage{microtype}
\usepackage[utf8x]{inputenc}
\usepackage[normalem]{ulem}

\usepackage{tikz}

\usepackage{amsthm}
\usepackage{amssymb}
\usepackage{amsmath}
\usepackage{mathrsfs}
\usepackage{xcolor}
\usepackage{enumerate}
\usepackage{smartref}
\usepackage{ucs,algorithm}
\usepackage[noend]{algpseudocode}

\newtheorem{theorem}{Theorem}
\newtheorem{proposition}[theorem]{Proposition}
\newtheorem{observation}[theorem]{Observation}

\newtheorem{definition}[theorem]{Definition}
\newtheorem{corollary}[theorem]{Corollary}
\newtheorem{lemma}[theorem]{Lemma}

\newcommand{\raus}[1]{}

\newenvironment{frcseries}{\fontfamily{frc}\selectfont}{}

\newcommand{\calF}{{\mathcal{F}}}

\newcommand{\calH}{{\mathcal{H}}}

\newcommand{\calP}{{\mathcal{P}}}

\newcommand{\calS}{{\mathcal{S}}}
\newcommand{\calT}{{\mathcal{T}}}

\newcommand{\var}[1]{\textsf{var}({#1})}

\newcommand{\sP}{\mathbf{\#P}}

\newcommand{\NP}{\mathbf{NP}}

\newcommand{\W}[1]{\mathbf{W[#1]}}
\newcommand{\FPT}{\mathbf{FPT}}

\newcommand{\CSPneg}{\mathrm{CSP_{neg}}}

\newcommand{\sCSP}{\mathrm{\#CSP}}
\newcommand{\sCSPneg}{\mathrm{\#CSP_{neg}}}

\newcommand{\twoSAT}{\mathrm{2\!-\!SAT}}
\newcommand{\stwoSAT}{\mathrm{\#2\!-\!SAT}}
\newcommand{\hornSAT}{\mathrm{Horn\!-\!SAT}}
\newcommand{\sSAT}{\mathrm{\#SAT}} 
\newcommand{\SAT}{\mathrm{SAT}} 

\PrerenderUnicode{é}

\begin{document}

\title{Hypergraph Acyclicity and Propositional Model Counting}
\author{Florent Capelli\thanks{ 
IMJ UMR 7586  -  Logique,
Université Paris Diderot,
France, Email: \texttt{fcapelli@math.univ-paris-diderot.fr}}
\and
Arnaud Durand\thanks{
IMJ UMR 7586  -  Logique,
Université Paris Diderot and LSV UMR 8643, ENS Cachan,
France, Email: \texttt{durand@math.univ-paris-diderot.fr}}
\and
Stefan Mengel\thanks{
 Laboratoire d'Informatique,
LIX UMR 7161,
Ecole Polytechnique,
France, Email: \texttt{mengel@lix.polytechnique.fr}. Partially supported by DFG grants BU 1371/2-2 and BU 1371/3-1.}
}
\maketitle

\begin{abstract}
We show that the propositional model counting problem $\sSAT$ for CNF-formulas with hypergraphs that allow a disjoint branches decomposition can be solved in polynomial time. We show that this class of hypergraphs is incomparable to hypergraphs of bounded incidence cliquewidth which were the biggest class of hypergraphs for which $\sSAT$ was known to be solvable in polynomial time so far. Furthermore, we present a polynomial time algorithm that computes a disjoint branches decomposition of a given hypergraph if it exists and rejects otherwise. Finally, we show that some slight extensions of the class of hypergraphs with disjoint branches decompositions lead to intractable $\sSAT$, leaving open how to generalize the counting result of this paper.
\end{abstract}

\section{Introduction}

Proposition model counting $(\sSAT)$ is the problem of counting satisfying assignments (models) to a CNF-formula. It is the canonical $\sP$-hard counting problem and is important due to its applications in Artificial Intelligence. Unfortunately, $\sSAT$ is extremely hard to solve: Even on restricted classes of formulas like monotome 2CNF-formulas or Horn 2CNF-formulas it is $\NP$-hard to approximate within a factor of $2^{n^{1-\epsilon}}$ for any $\epsilon > 0$~\cite{Roth96}. Fortunately, this is not the end of the story: While syntactical restrictions on the types of allowed clauses do not lead to tractable counting, there is a growing body of work that successfully applies so-called \emph{structural} restrictions to $\sSAT$, see e.g.~\cite{FMR-08,Samer:2010wc,PaulusmaSS13,SlivovskyS13}. In this line of work one does not restrict the individual clauses of CNF-formulas but instead the interaction between the variables in the formula. This is done by assigning graphs or hypergraphs to formulas and then restricting the class of (hyper)graphs that are allowed for instances (see Section \ref{sct:prelim} for details). In this paper we present a new class of hypergraphs, such with disjoint branches decompositions \cite{Duris12}, for which $\sSAT$ is tractable.

Having a disjoint branches decomposition is a so-called acyclicity notion for hypergraphs. Unlike for graphs, there are several resonable ways of  defining acyclicity for hypergraphs~\cite{Fagin-83} which have been very successful in database theory. Mostly three ``degrees of acyclicity'' have been studied: $\alpha$-acyclicity, $\beta$-acyclicity and $\gamma$-acyclicity, where the $\alpha$-acyclic hypergraphs form the most general and the $\gamma$-acyclic hypergraphs the least general class. Prior to this paper it was known that $\sSAT$ for CNF-formulas with $\alpha$-acyclic hypergraphs was $\sP$-hard \cite{Samer:2010wc}, while it is tractable for $\gamma$-acyclic hypergraphs as the latter have incidence cliquewidth bounded by $3$ \cite{GottlobP01} and thus the results of~\cite{SlivovskyS13} apply. 

To understand the influence of hypergraph acyclicity on the complexity of $\sSAT$, the next natural step is thus analyzing the intermediate case of $\beta$-acyclic hypergraphs. For this class it is known that $\SAT$ is tractable \cite{OrdyniakPS13}, unlike for $\alpha$-acyclic hypergraphs. Unfortunately, the algorithm in \cite{OrdyniakPS13} is based on a resolution-like method and it is not clear whether one can obtain tractability for counting from the method used for decision. In fact, most classical decision results based on tractability of resolution (such as for  $\twoSAT$)  or unit propagation ($\hornSAT$) do not extend to counting as the respective counting problems are hard (see e.g.~\cite{Roth96}).

Unfortunately, $\sSAT$ for CNF-formulas with $\beta$-acyclic hypergraphs has turned out to be a stubborn problem whose complexity could so far not be determined despite considerable effort by us and others \cite{Slivovsky14}. A natural approach which we follow in this paper is thus trying to understand slightly more restrictive notions of acyclicity. We focus here on hypergraphs with disjoint branches decompositions, a notion which was introduced by Duris~\cite{Duris12} and which lies strictly between $\beta$-acyclicity and $\gamma$-acyclicity. We show that for CNF-formulas whose hypergraphs have a disjoint branches decompositions we can solve $\sSAT$ in polynomial time. We also show that hypergraphs with disjoint branches decompositions are incomparable to hypergraphs with bounded incidence cliquewidth which so far were the biggest class of hypergraphs for which $\sSAT$ was known to be tractable. 
 Thus our results give a new class of tractable instances for $\sSAT$, pushing back the known tractability frontier for this problem.

Our main contribution is twofold: Most importantly, we present the promised counting algorithm for CNF-formulas whose hypergraphs have a disjoint branches decomposition in Section \ref{sct:counting}. Secondly, we present in Section \ref{sct:computedecomps} a polynomial time algorithm that checks if a hypergraph has a disjoint branches decomposition and if so also constructs it. On the one hand, this gives some confidence that hypergraphs with disjoint branches decompositions form a well-behaved class as it can be decided in polynomial time. On the other hand, the counting algorithm will depend on knowing a decomposition, so its computation is an essential part of the counting procedure.
Finally, in Section \ref{sct:extensions} we then turn to generalizing the results of this paper, unfortunately showing only negative results. We consider some natural looking extensions of hypergraphs with disjoint branches and show that $\sSAT$ is intractable on these classes under standard complexity theoretic assumptions.

\section{Preliminaries and notation}\label{sct:prelim}

\subsection{Hypergraphs and graphs associated to CNF-formulas}\label{sct:hypergraphintro}

In this section we describe graphs and hypergraphs commonly associated to CNF-formulas and introduce restricted classes of hypergraphs that we will consider in this paper. The \emph{primal graph} of a CNF-formula $F$ has as vertices the variables of $F$ and two vertices are connected by an edge if they appear in a common clause of $F$. The \emph{incidence graph} of $F$ is defined as the bipartite graph which has as vertices the variables and the clauses of $F$ and two vertices $u$ and $v$ are connected by an edge if $u$ is a variable and $v$ is a clause such that $u$ appears in $v$. The \emph{signed incidence graph} of such a formula is obtained from its incidence graph by orientating edges to indicate positive or negative occurences of variables in clauses (see~\cite{FMR-08} for details). 

A (finite) hypergraph  $\calH$ is a pair $(V,E)$ where $V$ is a finite set and $E\subseteq \calP(V)$. A subhypergraph $\calH'=(V',E')$ of $\calH=(V,E)$ is a hypergraph with $V'\subseteq V$ and $E'\subseteq \{e \cap V' \mid e \in E, e \cap V' \ne \emptyset\}$. A path between two vertices $u,v\in V$ is defined to be a sequence $e_1, \ldots, e_k$ such that $u\in e_1$, $v\in e_k$ and for every $i=1,\ldots , k-1$ we have $e_i\cap e_{i+1}\ne \emptyset$. A hypergraph $\calH$ is called \emph{connected} if there is a path between every pair of vertices of $\calH$. A (connected) \emph{component} of $\calH$ is defined to be a maximal connected subhypergraph of $\calH$.

To a CNF-formula $F$ we associate a hypergraph $\calH=(V,E)$ where $V$ is the variable set of  $F$ and the hyperedge set $E$ contains for each clause of $F$ an edge containing the variables of the clause.


\subsubsection{Graph Decompositions}

We will not recall basic graph decompositions such as tree-width and clique-width (see e.g.~\cite{GottlobP01,FMR-08}). A class of CNF-formulas is defined to be of bounded  (signed) incidence clique-width, if  their (signed) incidence graphs are of bounded clique-width. A set $X\subseteq V$ of vertices of a graph is called a \emph{module}, if every $v\in V\backslash X$ has the same set of neighbours and non-neighbours in $X$. Intuitively, the elements of a module $X$ are indiscernible by vertices outside of $X$.  If $X$ is a module of a graph $G$, the graph $G'=(V',E')$ obtained after contraction of $X$ is defined by $V':=(V\backslash X)\cup \{x\}$ where $x$ is a new vertex not in $V$ and $E':=(E\cap (V\backslash X)^2)  \cup \{ux: u\not\in X \mbox{  and } \exists  v\in X \mbox{ s.t. } uv\in E\})$. A class of CNF-formulas is of bounded \emph{modular incidence treewidth} if their incidence graphs are of bounded tree-width after contracting all modules.

\subsubsection{Acyclicity in hypergraphs}

It is well-known that, in contrast to the graph setting, there are several non equivalent notions of acyclicity for hypergraphs~\cite{Fagin-83}. Most of these notions have many equivalent definitions (see~\cite{Fagin-83,Duris12} for elimination rule based or cycle based definitions, for example), but we will mainly restrict ourselves to definitions using the notion of join trees. 

\begin{definition}\label{def:jointree} A \emph{join tree} of a hypergraph $\calH=(V,E)$ is a pair $(\calT,\lambda)$ where $\calT=(N,T)$ is a tree and $\lambda$ is a bijection between $N$ and $E$ such that:
\begin{itemize}
\item for each $e\in E$, there is a $t\in N$ such that $\lambda(t)=e$, and
\item for each $v\in V$, the set $\{t\in N\mid v \in \lambda(t)\}$ is a connected subtree of $\calT$.
\end{itemize}
\end{definition}

The second condition in Definition \ref{def:jointree} is often called the \emph{connectedness condition}. It is often convenient to identify an edge $e\in E$ with the vertex $\lambda(e)\in N$ and we will mostly follow this convention in this paper. We call a join tree $(\calT, \lambda)$ a \emph{join path} if the underlying tree~$\calT$ is a path.

A hypergraph is defined to be $\alpha$-\emph{acyclic} if it has a join tree~\cite{Fagin-83}.  This is the most general acyclicity notion for hypergraphs commonly considered. However, $\alpha$-acyclicity is  not closed under taking subhypergraphs: an $\alpha$-acyclic hypergraph may have cyclic subhypergraphs. To remedy this situation, one considers the restricted notion of $\beta$-\emph{acyclicity} where a hypergraph is defined to be $\beta$-acyclic if it is $\alpha$-acyclic and all of its subhypergraphs are also all $\alpha$-acyclic. 

A $\gamma$-cycle in a hypergraph is a sequence $(e_1,x_1,...,e_n,x_n)$ with  $n\geq 3$  where the $x_i$ are distinct vertices and the $e_i$ are distinct hyperedges such that, 
\begin{itemize}
\item for all $i\in [1,...,n-1]$, $x_i$ belongs to $e_i$ and $e_{i+1}$ and to no other $e_j$ for  $j\neq i,i+1$.
\item $x_n$ belongs to $e_n$ and $e_1$ and to possibly to other $e_j$s.
\end{itemize}

A hypergraph is $\gamma$-\emph{acyclic} if it has no $\gamma$-cycle. This notion can also be characterized and generalized through the notion of disjoint branches decompositions.

\begin{definition}
A \emph{disjoint branches decomposition} of a hypergraph $\calH$ is a join tree $(\calT,\lambda)$ such that for every two nodes $t$ and $t'$ appearing on different branches of $\calT$ we have $\lambda(t) \cap \lambda(t') = \emptyset$.   
 \end{definition}

Disjoint branches decompositions were introduced by Duris~\cite{Duris12} who proved that a hypergraph is $\gamma$-acyclic if and only if it has a disjoint branches decomposition for any choice of hyperedge as a root. Furthermore, he showed that every hypergraph with a disjoint branches decomposition is $\beta$-acyclic.

\subsection{Known complexity results and comparisons between classes}

We show the known complexity results for the restrictions of $\sSAT$ we have introduced before in Table~\ref{table:results}; for definitions of the appearing complexity classes see e.g.~\cite{FlumG06}.

\begin{table}\label{table:results}
\begin{center}
\begin{tabular}{l|l|l}
class & lower bound & upper bound\\
\hline
primal treewidth &  & $\FPT$~\cite{Samer:2010wc}\\
incidence treewidth &  & $\FPT$~\cite{Samer:2010wc}\\
modular incidence treewidth &  & $\FPT$~\cite{PaulusmaSS13}\\
signed incidence cliquewidth &  & $\FPT$~\cite{FMR-08}\\
incidence cliquewidth & $\W{1}$-hard~\cite{OrdyniakPS13} & $\mathbf{XP}$~\cite{SlivovskyS13}\\
$\gamma$-acyclic & & $\mathbf{FP}$~\cite{GottlobP01,SlivovskyS13}\\
$\beta$-acyclic & ? & ? \\
$\alpha$-acyclic & $\sP$-hard~\cite{Samer:2010wc} & $\sP$\\
disjoint branches & & $\mathbf{FP}$ (this paper)
\end{tabular}
\end{center}
\caption{Known complexity results for structural restrictions of $\sSAT$.}
\end{table}

The four acyclicity notions and classes defined by bounding the introduced width measures form a hierarchy for inclusion which is depicted in Figure~\ref{FIG:hierarchy}. Most of the proofs of inclusion can be found in \cite{Fagin-83,Duris12,GottlobP01,PaulusmaSS13} and the references therein. We give the missing results in this sections. 


\paragraph*{$\gamma$-acyclicity and modular treewidth are incomparable.}
 We exhibit a family of $\gamma$-acyclic hypergraph whose associated incidence graph have unbounded modular treewdith. 
Let $n\in \mathbb{N}$ and $\calH_n$ be the hypergraph of vertex set $\{x_1,...,x_n,y_1,...,y_n\}$ and of hyperedge set
\[\{\{y_i,x_1,...,x_n\} \mid i\leq n\} \cup \{\{x_i\}\mid i\leq n\}.\]

Clearly, $\calH_n$ is $\gamma$-acyclic. Also, the incidence graph of $\calH_n$ has no modules because of the $y_i$ and the singleton edges $\{x_i\}$. Thus the treewidth of $\calH_n$ and its modular treewidth coincide. Furthermore, it is easy to see that the incidence graph of $\calH_n$ contains a subgraph that is ismorphic to $K_{n,n,}$. Since treewidth is stable under taking subgraphs and $K_{n,n}$ is well-known to have treewidth $n$, it follows that the hypergraphs $\calH_n$ have unbounded modular treewidth.

For the other direction, cycles have bounded modular treewidth but are not $\gamma$-acyclic.

 \paragraph*{Disjoint branches and incidence clique-width are incomparable.}
In this section we will show that unlike $\gamma$-acyclic hypergraphs the hypergraphs with disjoint branches decompositions have unbounded cliquewidth. In fact we will even show this for hypergraphs with join paths. Since join paths do not branch, these hypergraphs are a subclass of the hypergraphs with disjoint branches decompositions.
 
We will use the following characterization of hypergraphs with join paths.

 \begin{lemma}
 \label{lemma:pathorder}
 A hypergraph $\calH = (V,E)$ has a join path if and only
 if there exists an order~$<_E$ on the edge set $E$ of $\calH$ such that for all
 $e,f,g \in E$ such that $e <_E f <_E g$, if $v \in e \cap g$ then $v
 \in f$.
 \end{lemma}
 \begin{proof}
 If $\calP$ is a join path
  of $\calH$, we choose an orientation of this path
 and then define $e <_E f$ if and only if $e$ appears before $f$ in
 $\calP$. If $e <_E f <_E g$ and $v \in e \cap g$, then as $f$ is
 between $e$ and $g$ in $\calP$. From the connectedness condition of $v$, we get $v \in
 f$.
 
 For the other direction, let $<_E$ be an order on $E$. Let $E:=\{e_1,\ldots,e_n\}$ with $e_i <_E e_{i+1}$ for $i <
 n-1$. Let $\calP$ be the path whose vertices are $E$ and edges are
 $(e_i,e_{i+1})$ for $i < n$. We claim that $\calP$ is a join path of
 $\calH$. Obviously $\calP$ is a path, so we only have to show the connectedness property.
Let $v \in V$, then for all $i \leq k \leq j$, if $v \in e_i \cap
   e_j$, then $v \in e_k$ by assumption on $<_E$. So the edges containing $v$ are connected
   in $\calP$ which proves the claim.
 \end{proof}
 
\begin{definition}
Let $G = (X,Y,E)$ be a bipartite graph. A {\em strong ordering} $(<_X, <_Y)$ of $G$ is a pair of orderings on $X$ and $Y$ such that for all $x,x' \in X$ and $y,y' \in Y$, such that $x <_X x'$ and $y <_Y y'$, if $(x,y) \in E$ and $(x',y') \in E$, then $(x,y') \in E$ and $(x',y) \in E$. $G$ is called a {\em bipartite permutation graph} if it admits a strong ordering.
\end{definition}

Brandstädt and Lozin showed the following property of bipartite permutation graphs.

\begin{lemma}[\cite{brandloz}] \label{lem:brandstaedt}Bipartite permutation graphs have unbounded cliquewidth.
 \end{lemma}
 
It turns out that hypergraphs with a bipartite permutation incidence graph admit a join path.
 
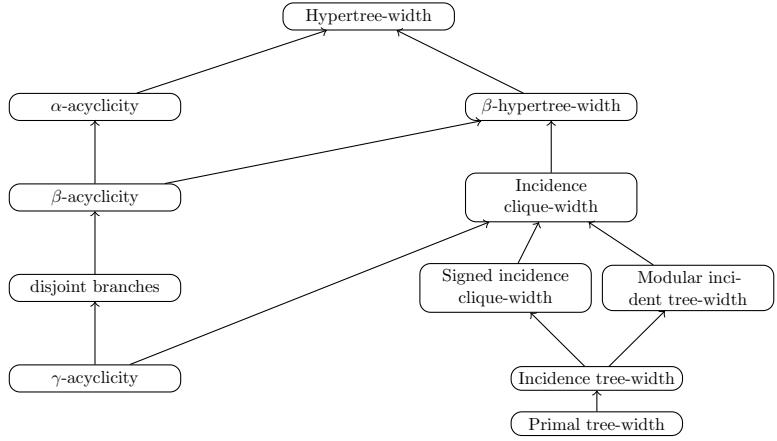
\begin{figure}
\center
\begin{tikzpicture}[scale=0.6, every node/.style={transform shape}, auto]
  \definecolor{sp}{RGB}{100,180,255}
  \definecolor{p}{RGB}{255,255,180}

  \tikzstyle{FPT} = [rectangle,draw,rounded corners,text centered,text width=3.5cm];
  \tikzstyle{XP} = [rectangle,draw,rounded corners,text centered,text width=3.5cm];
  \tikzstyle{sP} = [rectangle,draw,rounded corners,text centered,text width=3.5cm];
  \tikzstyle{P} = [rectangle,draw,rounded corners,text centered,text width=3.5cm];
  \tikzstyle{uK} = [rectangle,draw,rounded corners,text centered,text width=3.5cm];

  \node[P]  (ga)  at (0,0)  {$\gamma$-acyclicity};
  \node[P]  (db) at  (0,2)  {disjoint branches};
  \node[uK] (ba) at  (0,4) {$\beta$-acyclicity};
  \node[sP] (aa) at  (0,6) {$\alpha$-acyclicity};
 
  \node[sP] (hw) at  (6,8) {Hypertree-width};

  \node[FPT] (sicw) at (9,2) {Signed incidence clique-width};
  \node[XP] (icw) at (10,4) {Incidence clique-width};
  \node[sP] (bhw) at (10,6) {$\beta$-hypertree-width}; 
  \node[FPT] (mitw) at (13,2) {Modular incident tree-width};
  \node[FPT] (itw) at (11,0) {Incidence tree-width};
  \node[FPT] (ptw) at (11,-1) {Primal tree-width};

  \draw[->] (ga) -- (db);
  \draw[->] (db) -- (ba);
  \draw[->] (ba) -- (aa);
  \draw[->] (aa) -- (hw);
  \draw[->] (ga) -- (icw);
  \draw[->] (ptw) -- (itw);
  \draw[->] (itw) -- (sicw);
  \draw[->] (itw) -- (mitw);
  \draw[->] (mitw) -- (icw);
  \draw[->] (sicw) -- (icw);
  \draw[->] (icw) -- (bhw);
  \draw[->] (bhw) -- (hw);
  \draw[->] (ba) -- (bhw);
\end{tikzpicture}
\caption{A hierarchy of inclusion of graph and hypergraph classes. Classes not connected by a directed path are incomparable.}~\label{FIG:hierarchy}
\end{figure}

\begin{proposition}\label{prop:cliquewidthunbounded}
 Every hypergraph $\calH$  without empty hyperedges, whose incidence
 graph~$\calH_I$ is a bipartite permutation graph, has a join path.
 \end{proposition}
 \begin{proof}
 Let $(<_V,<_E)$ be a strong ordering of $\calH_I =
 (V,E,A)$. We claim that for all $e <_E f <_E
 g$, if $v \in e \cap g$ then $v \in f$. Indeed, as $f$ is not empty,
 there exists $w \in f$. If $w = v$, there is nothing to
 prove. Otherwise, suppose that $w <_V v$. Then, by definition of
 strong orderings, as $(f,w) \in A$ and $(g,v) \in A$, we have $(f,v)
 \in A$. Thus $v \in f$. The case $v <_V w$ follows symetrically:
 $(f,w) \in A$ and $(e,v) \in A$ implies that $(f,v) \in A$ so $v \in
 f$. Thus the ordering $<_E$ has the property of
 lemma~\ref{lemma:pathorder} and it follows that  $\calH$ has a join path.
 \end{proof}

By combining Lemma \ref{lem:brandstaedt} and Proposition \ref{prop:cliquewidthunbounded} we get:

\begin{corollary}
 The class of CNF-formulas with join paths has unbounded incidence cliquewidth. The same is thus true for CNF-formulas with disjoint branches decompositions.
\end{corollary}

Again, cycles give an example proving  the other direction of imcomparability.

\subsection{Representation of $\sSAT$ by constraint satisfaction problems}\label{sct:representation}

It will be convenient to describe our counting algorithm in the framework of constraint satisfaction problems in negative representation \cite{CohenGH2009}. We will discuss below how this representation relates to $\sSAT$. 

Let $D$ be a finite set called domain. A constraint $C=(R, \bar{x})$ is a pair where $R\subseteq D^{r}$ is a relation and $\bar{x}=(x_{i_1}, \ldots, x_{i_r})$ is a list of variables of length $r$. 
An instance $\Phi$ of the (uniform) constraint satisfaction problem, is a set of constraints. We denote by $\var{\Phi}$ the set $X$ of variables of $\Phi$. The instance $\Phi$ is satisfied by an assignment $a:\var{\Phi} \rightarrow D$ if for all  $(R,\bar x)\in \Phi$ we have $a(\bar x)=(a(x_{i_1}), \ldots, a(x_{i_r}))$ is in the relation $R$.  We denote this by $a\models \Phi$. 
 
The associated counting problem, $\sCSP$, is, given an instance $\Phi$, to compute 
\[ |\{a \mid a\models \Phi\}|, \] 
i.e., the number of satisfying assignments of $\Phi$.

No hypothesis is made above on the arity of relations which is not a priori bounded and may differ for different relations. So, it may be more succinct to represent each relation $R$ by listing the tuples in its complement $R^c:=D^r\backslash R$. Consequently, we define the counting constraint satisfaction problem in negative representation, $\sCSPneg$, that is to compute, given  $\Phi$ where each relation $R$ is encoded by listing the tuples in $R^c$, the number of satisfying assignments of $\Phi$.
 

\paragraph{The relation to $\sSAT$:} It is natural to represent CNF-formulas by a Boolean $\CSPneg$-instance. Indeed, as we do not bound the length $k$ of clauses, it is more realistic to represent each associated constraint relation not by its set of $2^{k}-1$ models (as common in the area of constraint satisfaction) but by its complement containing the unique counter model of the clause.

In the other direction one can easily encode every $\sCSPneg$-instance by a CNF-formula (see also \cite{BraultBaron12}): In a first step encode all domain elements in binary, introducing vertex modules in the hypergraph. Then we encode every tuple in every relation by a clause  that disallows the respective tuple. Observe that after the contraction of some modules, the CNF-formula has the same hypergraph as the original $\sCSPneg$-instance. Since the class of hypergraphs with disjoint branches decompositions is stable under introducing or contracting of modules, it follows that $\sSAT$ and $\sCSPneg$ are equivalent for our considerations.

\section{Counting solutions of disjoint branches queries}\label{sct:counting}

In this section we will show that $\sCSPneg$---and thus also $\sSAT$---restricted to hypergraphs with a disjoint branches decomposition can be solved in polynomial time. It will be convenient to work with inputs of the form $\phi = R_1 \lor \ldots \lor R_k$, i.e., we will count assignments $a$ that satisfy at least one of the $R_i$. By basic Boolean algebra and inclusion-exclusion, solving instances of this type is equivalent to $\sCSPneg$. Observe that when transforming a $\sCSPneg$-instance into the disjunctive form, this changes the encoding of the relations from negative representation to positive representation, i.e., we list the tuples in the relations $R_i$ explicitly and \emph{not} those in the complement.

Note that, as discussed in Section~\ref{sct:representation}, we may assume that the domain of all relations is $\{0,1\}$, so we restrict ourselves to this case.

Let us introduce some notation:
Let $X,Y$ be two sets of variables and let $a\colon X\rightarrow \{0,1\}$ and $b\colon Y \rightarrow \{0,1\}$ be two assignments. We call $a$ and $b$ \emph{consistent}, symbol $a\sim b$, if they agree on their common variables $X\cap Y$.
If, in addition, $X\subseteq Y$, we write $a\subseteq b$.  Finally, if $a$ and $b$ have disjoint domains, i.e., $X\cap Y = \emptyset$, we denote by $a \oplus b$ the assignment on $X\cup Y$ defined in the obvious way.

For an assignment $a\colon X\mapsto \{0,1\}$ and a set of variables $Y\subseteq X$ the restriction of $a$ onto $Y$ is denoted by $a|_Y$.
 
Let $\phi$ be an instance and let $\psi$ be a subformula of $\phi$. Let $X$ a subset of the variables such that $\var{\psi} \subseteq X \subseteq \var{\phi}$ and let $a$ be a partial assignment of variables in $\var{\phi}$. We denote by $Sol_X(\psi, a) = \{ b : X \rightarrow \{0,1\}, b \models \psi, a \sim b \}$ and by $S_X(\psi,a) = |Sol_X(\psi,a)|$. The number of solutions of $\phi$ is then $S_{\var{\phi}}(\phi,\emptyset)$, where $\emptyset$ denotes the empty partial assignment. We show that by computing a polynomial number of values $S_X(\psi,a)$ in polynomial time, we can compute $S_{\var{\phi}}(\phi,\emptyset)$. To this end we will prove several lemmas that will allow us a recursive computation.
 
The first lemma shows how the disjointness naturally appears when we want to count solutions:
\begin{lemma}\label{lem:basiccount}
Let $\phi_1$ and $\phi_2$ be two instances and $X\supseteq \var{\phi_1} \cup \var{\phi_2}$, then \[S_X(\phi_1 \lor \phi_2, a) = S_X(\phi_1,a) + S_X(\phi_2,a) - |Sol_X(\phi_1,a) \cap Sol_X(\phi_2,a)|.\]
\end{lemma}
\begin{proof}
Using inclusion-exclusion and the fact that $Sol_X(\phi_1 \lor \phi_2, a) = Sol_X(\phi_1,a) \cup Sol_X(\phi_2,a)$ directly yields the result.
\end{proof}
 
The next lemma will allow us to efficiently compute $Sol_X(\phi_1 \lor \phi_2, a)$ recursively.
 
 \begin{lemma}
   \label{lem:fusion}
   Let $X_1 = \var{\phi_1}$ and $X_2 = \var{\phi_2}$. Assume that $X_1 \cap X_2 = \emptyset$ and let $X = X_1 \cup X_2$. Let $a$ be a partial assignment of the variables of $X$ and $a_1 = a|_{X_1}$, $a_2 = a|_{X_2}$. Then $Sol_X(\phi_1,a) \cap Sol_X(\phi_2,a) = \{ b_1 \oplus b_2 \mid b_i : X_i \rightarrow \{0,1\}, b_i \models \phi_i, a_i \subseteq b_i\}$ and $|Sol_X(\phi_1,a) \cap Sol_X(\phi_2,a)| = S_{X_1}(\phi_1,a_1)S_{X_2}(\phi_2,a_2)$.
 \end{lemma}
   \begin{proof}
     The inclusion from right to left is trivial. For the other inclusion, it is enough to remark that for $b \in Sol_X(\phi_1,a) \cap Sol_X(\phi_2,a)$, we have $b = b|_{X_1} \oplus b|_{X_2}$ and $a_i \subseteq b|_{X_i}$ as $a \subseteq b$. The equality of the sizes follows directly. 
   \end{proof}
 
We now show how we can add variables that do not appear in $\phi$.
 
 \begin{lemma}
   \label{lem:extend}
   Let $Y \subseteq X$ and $a:X_0\rightarrow \{0,1\}$ for $X_0\subseteq X$. Then \[S_X(\phi,a) = 2^{|X \setminus (Y\cup X_0)|}S_Y(\phi,a|_{Y}).\]
 \end{lemma}
   \begin{proof}
     Every solution of $\phi$ on $Y$ can be arbitrarily extended on $X\setminus (X_0 \cup Y)$ and thus the claim follows directly.
   \end{proof}

The next corollary lets us handle the disjunction of more than two terms.

 \begin{corollary}
 \label{cor:fusion}
   Let $\phi_1, \ldots, \phi_k$ be formulas with $X_i = \var{\phi_i}$ and $X_i \cap X_j = \emptyset$ for every combination $i,j\in [k], i\ne j$. Let $X$ be a set such that $X_1 \cup \ldots \cup X_k \subseteq X$ and let $a:X_0\rightarrow \{0,1\}$ for $X_0\subseteq X$, $a_i = a|_{X_i}$. Then  
 $$S_X(\bigvee_{j=1}^k \phi_j, a) = 2^{|X \setminus (X_0 \cup X_1 \cup \ldots \cup X_k)|} \sum_{i=1}^k S_{X_i}(\phi_i,a_i) \prod_{j=1}^{i-1} (2^{|X_j \setminus X_0|}-S_{X_j}(\phi_j,a_j))  \prod_{j=i+1}^{k}2^{|X_j \setminus X_0|}$$
  \end{corollary}
 \begin{proof}
   The proof proceeds by induction on $k$. For $k=1$, we have to show $S_X(\phi_1,a)=2^{|X \setminus (X_0 \cup X_1)|}S_{X_1}(\phi_1,a_1)$ which is Lemma~\ref{lem:extend}.
 
   Now let $Z = X_1 \cup \ldots \cup X_{k+1}$ and $Z_1 = Z \setminus X_1$. To ease notations, we denote by $a'$ the assignment $a|_Z$. Remark that $a_i = a|_{X_i} = a'|_{X_i}$ and $a|_{Z_1} = a'|_{Z_1}$ as $X_i \subseteq Z$.

As $X_1$ is disjoint from $Z_1$, we can apply Lemma~\ref{lem:fusion} to $\phi_1$ and $\phi_2 \lor \ldots \lor \phi_{k+1}$ to get
 $$S_Z(\bigvee_{j=1}^{k+1} \phi_j,a') = S_Z(\phi_1,a') + S_Z(\bigvee_{j=2}^{k+1} \phi_j,a')-S_{X_1}(\phi_1,a_1)S_{Z_1}(\bigvee_{j=2}^{k+1} \phi_j,a|_{Z_1}).$$
 By applying Lemma~\ref{lem:extend} and distributivity, we get
 $$ S_Z(\bigvee_{j=1}^{k+1} \phi_j,a') = 2^{|Z_1 \setminus X_0|}S_{X_1}(\phi_1,a_1)+S_{Z_1}(\bigvee_{j=2}^{k+1}\phi_j,a|_{Z_1})(2^{|X_1 \setminus X_0|}-S_{X_1}(\phi_1,a_1)).$$
 Now we use the induction hypothesis on $\bigvee_{j=2}^{k+1}\phi_j$ and the equality $2^{|Z_1 \setminus X_0|} = \prod_{j=2}^{k+1} 2^{|X_j \setminus X_0|}$ to get 
 \begin{align*}  &  S_Z(\bigvee_{j=1}^{k+1} \phi_j,a') \\= 
   & \quad \prod_{j=2}^{k+1} 2^{|X_j \setminus X_0|} S_{X_1}(\phi_1,a_1) \\& \hspace*{1cm} +(2^{|X_1 \setminus X_0|}-S_{X_1}(\phi_1,a_1))\left(\sum_{i=2}^{k+1} S_{X_i}(\phi_i,a_i) \prod_{j=2}^{i-1} (2^{|X_j \setminus X_0|}-S_{X_j}(\phi_j,a_j)) \prod_{j=i+1}^{k+1} 2^{|X_j \setminus X_0|}\right)\\
 = & \sum_{i=1}^{k+1} S_{X_i}(\phi_i,a_i) \prod_{j=1}^{i-1} (2^{|X_j \setminus X_0|}-S_{X_j}(\phi_j,a_j))  \prod_{j=i+1}^{k+1}2^{|X_j \setminus X_0|}.\end{align*}
 
 Finally, applying Lemma~\ref{lem:extend} yields
 $$S_X(\bigvee_{j=1}^{k+1} \phi_j, a) = 2^{|X \setminus (X_0 \cup X_1 \cup \ldots \cup X_{k+1})|} \sum_{i=1}^{k+1} S_{X_i}(\phi_i,a_i) \prod_{j=1}^{i-1} (2^{|X_j \setminus X_0|}-S_{X_j}(\phi_j,a_j))  \prod_{j=i+1}^{k+1}2^{|X_j \setminus X_0|} .$$
  which is the desired result.
 \end{proof}
 
 
 A final lemma will help us to compute the size of the intersections of the solutions to a formula and a single relation.
 
 \begin{lemma}
 \label{lem:sum}
   $|Sol_X(R,a) \cap Sol_X(\phi,a)| = \sum_{b \in Sol_X(R,a)} S_X(\phi,b)$
 \end{lemma}
 \begin{proof}
   This follows from the fact that $Sol_X(R,a) \cap Sol_X(\phi,a) = \bigcup_{b \in Sol_X(R,a)} Sol_X(\phi,b)$ and that the union is disjoint.
 \end{proof}
 
We now finally show the main result of this section.
 
\begin{theorem}
There is a polynomial time algorithm that, given an instance $\phi = \bigvee_{i=1}^m R_i$ and a disjoint branches decomposition of the hypergraph of $\phi$, computes the number of satisfying assignments of $\phi$.
\end{theorem}
\begin{proof}
Let $\phi= R_1 \lor \ldots \lor R_m$ be an instance with hypergraph $\calH$. Let $r_i := |R_i|$ and let $r:=\sum_{i=1}^m r_i$. Let furthermore $(\calT,\lambda)$ be a disjoint branches decomposition of $\calH$. For a vertex $t$ of $\calT$, we denote by $\calT_t$ the subtree of $\calT$ rooted in $t$ and by  $\phi_t$ the associated subinstance. Furthermore, $R_t$ is defined to be the relation associated to $t$. Finally, we denote be $V_t$ the set of variables of $\phi_t$. 

We will give a polynomial time algorithm that computes inductively from the leaves to the root of $\calT$ certain values $S_X(\psi,a)$ where $\psi$ is a subinstance of $\phi$. Our goal is to compute $S_{V_r}(\phi_r,\emptyset)$ where $r$ is the root of $\calT$, since this value is the number of solutions of $\phi$.
More precisely, for a given $t$ in the tree, we compute $S_{V_t}(\phi_t, \emptyset)$ and for all ancestors $u$ of $t$ and all $b \in R_u$, we compute $S_{V_t}(\phi_t,b|_{\var{R_t}})$. Since there are $m$ vertices in $\calT$ and at most $r+1$ values to compute for each vertex, we will to compute at most $m(r+1)$ different values. We will show how to compute these values in polynomial time to get a polynomial time algorithm overall.
 
Let $t$ be a vertex of $\calT$, $u$ one of his ancestors and $b \in R_u$. If $t$ is a leaf, then $\phi_t$ consists only of the relation $R_t$. Then $Sol_{V_t}(\phi_t, \emptyset)= R_t$ and $Sol_{V_t}(\phi_t,b|_{\var{R_t}})=\{ a\in R_t\mid a\sim b\}$ so the computations can be done efficiently.
 
Now that assume $t$ has children $t_1,\ldots,t_k$. To ease notation, let $V_i = V_{t_i}$, $\phi_i = \phi_{t_i}$ and $R_i = R_{t_i}$. Observe that by Lemma~\ref{lem:basiccount} we have
$$ S_{V_t}(\phi_t,\emptyset) = S_{V_t}(R_t,\emptyset) + S_{V_t}(\bigvee_{j=1}^k \phi_j, \emptyset) - |Sol_{V_t}(R_t,\emptyset) \cap Sol_{V_t}(\bigvee_{j=1}^k \phi_j, \emptyset)|.$$
As the variables of the $\phi_i$ are disjoint, by Corollary~\ref{cor:fusion}, one can compute $S_{V_t}(\bigvee_{j=1}^k \phi_j, \emptyset)$ in $O(k)$ if the values of $Sol_{V_i}(\phi_i,\emptyset)$ are precomputed, which is the case by induction.
 
In addition, $S_{V_t}(R_t,\emptyset) = 2^{|V_t \setminus \var{R_t}|} |R_t|$ since a solution of $R_t$ on variables $V_t$ is a solution of $R_t$ on $\var{R_t}$ and any assignment of the other variables.
 
Finally, $|Sol_{V_t}(R_t,\emptyset) \cap Sol_{V_t}(\bigvee_{j=1}^k \phi_j, \emptyset)| = \sum_{a \in R_t} S_{V_t}(\bigvee_{j=1}^k \phi_j, a)$ by Lemma~\ref{lem:sum}. By Corollary~\ref{cor:fusion}, one can compute for each $a$ the value $S_{V_t}(\bigvee_{j=1}^k \phi_j, a)$ in time $O(k)$ if the values of $S_{V_i}(\phi_i, a|_{V_i})$ are precomputed. But since the domain of $a$ is $\var{R_t}$, we have $a|_{V_i} = a|_{\var{R_t} \cap V_i} = a|_{\var{R_i}}$ by connectedness of the variables in the join tree $\calT$. Thus $S_{V_i}(\phi_i, a|_{V_i}) = S_{V_i}(\phi_i, a|_{\var{R_i}})$ which is precomputed by hypothesis.
 
Let $b' = b|_{\var{R_t}}$. We compute $S_{V_t}(\phi_t, b')$ in the following way, similarly to before. We start with Lemma~\ref{lem:basiccount} to get
$$ S_{V_t}(\phi_t,b') = S_{V_t}(R_t,b') + S_{V_t}(\bigvee_{j=1}^k \phi_j, b') - |Sol_{V_t}(R_t,b') \cap Sol_{V_t}(\bigvee_{j=1}^k \phi_j, b')|.$$
Again, by Corollary~\ref{cor:fusion}, one can compute $S_{V_t}(\bigvee_{j=1}^k \phi_j, b')$ in $O(k)$ if $S_{V_i}(\phi_i, b'|_{V_i})$ are known. But as the domain of $b'$ is $\var{R_t}$, $b'|_{V_i} = b'|_{\var{R_i}}$ by connectedness of the variables in $\calT$. So $S_{V_i}(\phi_i, b'|_{V_i})$ is precomputed since $u$ is also an ancestor of $t_i$.
 
Moreover, $S_{V_t}(R_t,b') = \sum_{a \in R_t, b' \subseteq a} 2^{|V_t \setminus \var{R_t}|}$ which can be computed in $O(|R_t|)$. 
 
Finally, by Lemma~\ref{lem:sum}, we have 
$$|Sol_{V_t}(R_t,b') \cap Sol_{V_t}(\bigvee_{j=1}^k \phi_j, b')| = \sum_{a \in R_t, b' \subseteq a} S_{V_t}(\bigvee_{j=1}^k \phi_j, a).$$
And again, by Corollary~\ref{cor:fusion}, we can compute $S_{V_t}(\bigvee_{j=1}^k \phi_j, a)$ in time $O(k)$ if $S_{V_i}(\phi_i, a|_{V_i})$ is precomputed. For the same reasons as above, $a|_{V_i} = a|_{\var{R_i}}$, thus these values were already computed by induction. 
 
To conclude, we have seen that one can compute the $S_{V_t}(\phi_t,\emptyset)$ and $S_{V_t}(\phi_t,b|_{\var{R_t}})$ for each $b\in R_u$ where $u$ is an ancestor of $t$ with $O(k\cdot r)$ arithmetic operations. Thus we can compute $S_{V_r}(\phi_r,\emptyset)$ in polynomial time.
\end{proof}

\section{Computing disjoint branches decompositions}\label{sct:computedecomps}

In this section we will show how to compute disjoint branches decompositions of hypergraphs in polynomial time. We will first introduce $PQF$-trees, the datastructure that our algorithm relies on, then consider some structural properties of hypergraphs with disjoint branches decompositions and finally describe the algorithm itself, relying on objects we call $A$-separators.

\subsection{$PQF$-trees}
\label{subsec:pqf}

$PQ$-trees are a data structure introduced by Booth and Lueker \cite{boothlueker} originally to check matrices for the so-called consecutive ones property. This problem can be reformulated as follows in our setting: Given a hypergraph $\calH=(V,E)$, is there an ordering $\ell = e_1 \ldots e_m$ of the edges such that if $v \in e_i \cap e_j$, then for all $i \leq k \leq j$, $v \in e_k$? We encode ordering of edges by lists. We call such a list {\em consistent} for $\calH$. Note that the notion of consistent lists matches exactly our notion of join paths. 

A $PQ$-tree is a compact way of representing all the consistent lists for a hypergraph. We introduce a generalization of this data structure which we call $PQF$-trees.

\begin{definition}
Let $\calH=(V,E)$ be a hypergraph. A {\em $PQF$-tree} for $\calH$ is defined to be an ordered tree with leaf set $E$ such that 
\begin{itemize}
 \item the internal nodes are labeled with $P$, $Q$ or $F$, 
 \item the $P$-nodes and $F$-nodes have at least two children, and
 \item the $Q$-nodes have at least $3$ children.
\end{itemize}
A $PQF$-tree without $F$-nodes is called a $PQ$-tree.
\end{definition}

$PQF$-trees will be used to encode sets of permutations of the edge set of a hypergraph that have certain properties. We write these permutations simply as (ordered) lists. To this end, we define some notation for lists and sets of lists. The concatenation of two ordered lists $\ell_1,\ell_2$ will be denoted by $\ell_1 \ell_2$. If $L_1,L_2$ are two sets of lists, we denote by $L_1L_2$ the set $\{\ell_1\ell_2 \mid \ell_1 \in L_1,\ell_2 \in L_2\}$. With this notation we now define the notion of the frontiers of a $PQF$-tree which will be central to this section.

\begin{definition}
The {\em frontiers} $\calF(T)$ of a $PQF$-tree $T$ for $\calH=(V,E)$ are a set of ordered list of the elements of $E$ defined inductively by
\begin{itemize}
\item if $T$ is a leaf $e$, then $\calF(T) = \{e\}$,
\item if $T$ is rooted in $t$, having children $t_1,\ldots,t_k$, then
  \begin{itemize}
  \item if $t$ is an $F$-node then $\calF(T) = \calF(T_1) \ldots \calF(T_k)$,
  \item if $t$ is a $Q$-node then $\calF(T) = (\calF(T_1) \ldots \calF(T_k)) \cup (\calF(T_k) \ldots \calF(T_1))$,
  \item if $t$ is a $P$-node then $\calF(T) = \bigcup_{\sigma \in \calS_k} \calF(T_{\sigma(1)}) \ldots \calF(T_{\sigma(k)})$ where $\calS_k$ is the set of permutations of $[k]$,
  \end{itemize}
  where $T_i$ is the subtree of $T$ rooted in $t_i$.
\end{itemize}

If for all $\ell \in \calF(T)$, $\ell$ is a consistent list for $\calH$, then we say that $T$ is \emph{consistent} for $\calH$.
\end{definition}
 
We recall the main theorem of \cite{boothlueker}, which allows to compute all possible join paths of a hypergraph in polynomial time.
\begin{theorem}[\cite{boothlueker}]
\label{th:boothluecker}
Given a hypergraph $(V,E)$, one can compute in time $O(|E||V|)$ a $PQ$-tree $T$ such that $\calF(T)$ is exactly the set of consistent lists for $\calH$.
\end{theorem}

In order to compute disjoint branches decompositions, we will need to compute join paths with additional restrictions. This is the reason for the introduction of $F$-nodes. It will be convenient to not have $F$-nodes that are children of other $F$-nodes and thus we introduce the following normal form for $PQF$-trees. 

\begin{definition}
A $PQF$-tree $T$ is said to be in {\em normal form} if there is no $F$-node in $T$ having an $F$-node as a child.
\end{definition}

Clearly, if an $F$-node $t$ has a child $u$ in $T$ which is also an $F$-node, then we can remove~$u$ from $T$ and connect its children to $t$ without changing $\calF(T)$. Thus we may always assume that all $PQF$-trees we encounter are in normal form.


We will in the remainder of this section use certain subtrees of $PQF$-trees which we call $PQF$-subtrees. These will be trees rooted in a vertex $t$ of a $PQF$-tree, but they will not necessarily contain all descendants of $t$. Instead, we allow to ``cut off'' certain trees that are rooted by children of $t$. We now give a formal definition of $PQF$-subtrees. As usual, the {\em subtree rooted in $t$} is defined to be the tree induced by $t$ and all its descendants.

\begin{definition}
Let $T$ be a $PQF$-tree, and let $t$ be a vertex of $T$. A subgraph $S$ of $T$ is said to be a \emph{$PQF$-subtree rooted in $t$} if
\begin{itemize}
\item $t$ is a leaf and $S$ consists of the graph containing only $t$,
\item $t$ is a $P$-node and $S$ is the subtree rooted in $t$, or 
\item $t$ is a $Q$-node or an $F$-node with children $t_1, \ldots, t_k$ and there exists $i,j$ such that $1 \leq i < j \leq k$ and $S$ is the graph containing $t$ and $T_i,\ldots,T_k$, the subtrees rooted in $t_i, \ldots, t_j$. 
\end{itemize}
\end{definition}

We will now show that $PQF$-subtrees allow us to filter the frontier of a $PQF$-tree for certain lists that we will be interested in later. Remember that the \emph{depth} of a node in a tree is its distance from the root.

\begin{lemma}
\label{lem:pqfsub}
Let $T$ be a consistent $PQF$-tree for $(V,E)$ in normal form. Let $V' \subseteq V$ and $A = \{e \in E \mid  V' \subseteq e \}$. Then there exists a $PQF$-subtree $T_{V'}$ of $T$ such that the labels of the leaves of $T_{V'}$ are exactly $A$.
\end{lemma}
\begin{proof}
First assume that $V' = \{v\}$. Let $t$ be the deepest node of $T$ such that the set $A$ is contained in the set of labels of the leaves of the subtree rooted in $t$. 

If $t$ is a leaf, then the $PQF$-subtree containing only $t$ is the subtree we are looking for. 

Otherwise, let $t_1, \ldots, t_k$ be the children of $t$. By maximality of the depth of $t$, we know that there are at least two children of $t$ such that the subtrees rooted in them contain elements of $A$ in their leaf labels. Let $t_i$ be the leftmost such child and $t_j$ the rightmost one. Note that $1 \leq i < j \leq k$. Let $T_0$ be the subtree of $T$ rooted in $t$ and $T_1,\ldots,T_k$ the subtrees rooted in $t_1,\ldots,t_k$. Furthermore, we choose $\ell_s \in \calF(T_s)$ for $s = 1,\ldots, k$ arbitrarily. Note that the leaves of $T_s$ are the edges in $l_s$.

If $t$ is a $P$-node then all the leaves of $T_0$ are in $A$. Indeed, suppose first that $i \neq 1$. As $t$ is a $P$-node, there exists a list $\ell$ in $\calF(T)$ having $\ell_i\ell_1\ell_j$ as a sublist. By definition of $i$, the leaves in $\ell_1$ do not contain $v$, but $v$ is contained in some edges in $\ell_i$ and $\ell_j$. Thus the connectedness condition for $v$ is not respected in $\ell$. This is a contradiction since $T$ is consistent for $\calH$. Analogously, we show $j = k$. Now, since $t$ is a $P$-node, the lists $\ell_1\ell_k$ and $\ell_k\ell_1$ are sublist of of lists in $\calF(T)$. Then because of the connectivity condition, the first and the last edges of $\ell_1$ and $\ell_k$ contain $v$. Thus all edges of $\ell_1$ and $\ell_k$ are in $A$ and consequently all edges of the subtree rooted in $t$.

If $t$ is an $F$-node, we show that $\ell_i$ contains only edges in $A$. By connectedness of $v$, the last element of $\ell_i$ is in $A$. Moreover, as $T$ is in normal form, $t_i$ is not an $F$-node, thus the list $\ell_i'$ obtained by reversing $\ell_i$ is in $\calF(T_i)$ and for the same reason, its last edges---that is the first of $\ell_i$---is in $A$. It follows that all edges in $\ell_i$ are in $A$. Analogously, all edges in $\ell_j$ are in $A$, and by connectedness of $v$, all edges in the leafs of $T_i,\ldots,T_j$ are in $A$. Thus the leaves of the $PQF$-subtree rooted in $t$ containing the subtrees $T_i, \ldots, T_j$ are exactly $A$, so this is the desired $PQF$-subtree.

If $t$ is a $Q$-node, then $\ell_i \ldots \ell_j$ and $\ell_j \ldots \ell_i$ are sublists of a list in $\calF(T)$.  With the connectivity condition it follows that that the first and the last edges of $\ell_i$ and $\ell_j$ are in $A$. Thus we find a $PQF$-subtree with the desired properties as before. This completes the case $V'=\{v\}$.

Now, if $V' = \{v_1, \ldots, v_p\}$ we construct $S$ iteratively. To this end, let $T_0 = T$ and for $i=1, \ldots p$ we let $T_{i}$ be the $PQF$-subtree of $T_{i-1}$ whose leaves are exactly the edges containing $v_{i+1}$. The tree $T_i$ can be found as argued above. Obviously, we have $T_{V'} = T_p$ whic completes the proof.
\end{proof}

During the construction of disjoint branches decompositions later, we will put restrictions on the position of some edges in join paths. To do so we will use the algorithm of the following proposition.

\begin{proposition}
\label{prop:force}
  There is a polynomial time algorithm \texttt{Force} that, given a $PQF$-tree $T$ and a $PQF$-subtree $S$ of $T$, computes in polynomial time a $PQF$-tree $T' = \text{\emph{\texttt{Force(T,S)}}}$ such that $\calF(T') = \{\ell_1\ell_2 \in \calF(T) \mid \ell_2 \in \calF(S) \}$. If this set is empty, the algorithm rejects.
\end{proposition}
\begin{proof}
Let $s$ be the root of the $PQF$-subtree $S$. We describe the algorithm by induction on the depth of $s$. Assume the depth of $s$ is $0$, that is if $s$ is the root of $T$. If $S = T$, then we simply set $T':= T$. Otherwise, first observe that $s$ is not a $P$-node since it would imply that $S = T$ by definition of $PQF$-trees. Let $s_1,\ldots,s_k$ be the children of $s$ and $1 \leq i < j \leq k$ such that $S$ is the $(i,j)$-$PQF$-subtree rooted in $s$.

If $s$ is an $F$-node, then $\calF(S) = \calF(T_i) \ldots \calF(T_j)$ and $\calF(T) = \calF(T_1) \ldots \calF(T_k)$. Thus if $j < k$, then $\{\ell_1\ell_2 \in \calF(T) \mid \ell_2 \in \calF(S) \}$ is empty since every $\ell \in \calF(T)$ is of the form $\ell'\ell_k$ with $\ell_k \in \calF(T_k)$ that is disjoint from the leaves in $S$. Thus the algorithm rejects in this case. However, if $j = k$ then $\calF(S) = \calF(T_i) \ldots \calF(T_k)$ and then for all $\ell \in \calF(T)$, $\ell = \ell_1\ell_2$ with $\ell_2 \in \calF(S)$ (and $\ell_1 \in \calF(T_1) \ldots \calF(T_{i-1})$). Consequently, $T$ is the desired $PQF$-tree.

If $s$ is a $Q$-node, $\calF(S) = \calF(T_i) \ldots \calF(T_j) \cup \calF(T_j) \ldots \calF(T_i)$ and $\calF(T) = \calF(T_1) \ldots \calF(T_k) \cup \calF(T_k) \ldots \calF(T_1)$. Thus, for the same reasons as in the last case, if $i \neq 1$ and $j \neq k$, there is no list in $\calF(T)$ of the form $\ell_1\ell_2$ with $\ell_2 \in \calF(S)$ and the algorithm rejects.

Suppose that $j = k$. Let $T'$ be the $PQF$-tree $T$ where we replace $s$ with an $F$-node. We have $\calF(T') = \calF(T_1) \ldots \calF(T_k)$. Just as before, we have $\calF(T') \subseteq \{\ell_1\ell_2 \in \calF(T) \mid \ell_2 \in \calF(S) \}$. For the other inclusion, as $S \neq T$, we also have $1 < i$. Let $\ell_1\ell_2 \in \calF(T)$ such that $\ell_2 \in \calF(S)$. Thus $\ell_1 \in \calF(T_1) \ldots \calF(T_{i-1})$ is not empty and thus $\ell_1\ell_2 \in \calF(T')$.

Now if $i = 1$, we reverse the children of $s$. This does not change $\calF(T)$ and this case reduces to the previous one.

Now suppose that the depth of $s$ is $d+1$. Let $r$ be the root of $T$ and $t$ the child of $r$ such that the subtree $T_t$ rooted in $t$ contains $S$. In $T_t$, the depth of $s$ is $d$. We recursively apply our algorithm on $(T_t,S)$ to get a new $PQF$-tree $T_t'$ such that $\calF(T_t') = \{\ell_1\ell_2 \in \calF(T_t) \mid \ell_2 \in \calF(S) \}$. We claim that if we apply apply the transformation shown in the Figure~\ref{fig:transf} to the root $r$ we get $T'$ with the desired properties.

If the root $r$ is a $Q$-node, for the same reasons as in the cases for depth $0$, if $t$ is not the rightmost or the leftmost child of $r$, then the transformation is impossible and the algorithm rejects. If not, we force the subtree $T_t'$ on the right of the $PQF$-tree and the same reasoning as before will give the desired result.

If $r$ is an $F$-node, the transformation is essentially the same as for $Q$-nodes except that we need $t$ to be the rightmost child of $r$, since we cannot reverse the children here.

Let now $r$ be a $P$-node. First observe that by induction hypothesis $\calF(T_t') = P \calF(S)$ for some set $P$. Thus the set $\{\ell_1\ell_2 \in \calF(T) \mid \ell_2 \in \calF(S) \}$ is obtained when we permute the children of $\calP$ and bring $\calT_t'$ on the right side. Thus we can apply any permutation on the $k-1$ other children and let $t$ on the right, which is exactly what the first transformation does.

Finally we see that we perform at most one change on each vertex lying between $r$ and $s$, so the construction can easily be done in polynomial time.
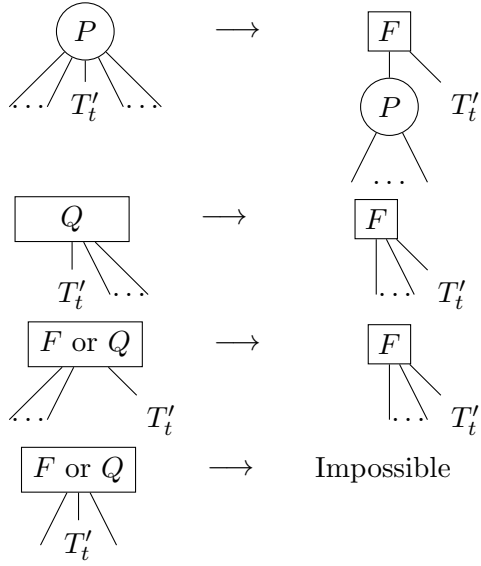
\begin{figure}
  \centering
    \begin{tikzpicture}
    \node[draw,circle] (1) at (0,0) {$P$};
    \node (2) at (0,-1) {$T_t'$};
    \draw (1) -- (2); 
    \draw (1) -- (0.5,-1);
    \node (3) at (0.75,-1) {$\ldots$};
    \draw (1) -- (1,-1);  
    \draw (1) -- (-0.5,-1);
    \node (3) at (-0.75,-1) {$\ldots$};
    \draw (1) -- (-1,-1);

    \node (0) at (2,0) {$\longrightarrow$};

    \node[draw,rectangle] (4) at (4,0) {$F$};
    \node (5) at (5,-1) {$T_t'$};
    \node[draw,circle] (6) at (4,-1) {$P$};
    \draw (4) -- (5);
    \draw (4) -- (6);
    \draw (6) -- (3.5,-2);
    \node (7) at (4,-2) {$\ldots$};
    \draw (6) -- (4.5,-2);
  \end{tikzpicture}

  \begin{tikzpicture}
    \node[draw,rectangle,minimum width=1.5cm] (1) at (0,0) {$Q$};
    \node (2) at (0,-1) {$T_t'$};
    \draw (1) -- (2); 
    \draw (1) -- (0.5,-1);
    \node (3) at (0.75,-1) {$\ldots$};
    \draw (1) -- (1,-1);  

    \node (0) at (2,0) {$\longrightarrow$};

    \node[draw,rectangle] (4) at (4,0) {$F$};
    \node (5) at (5,-1) {$T_t'$};
    \draw (5) -- (4);
    \draw (4) -- (4,-1);
    \node (6) at (4.25,-1) {$\ldots$}; 
    \draw (4) -- (4.5,-1);
  \end{tikzpicture}

  \begin{tikzpicture}
    \node[draw,rectangle,minimum width=1.5cm] (1) at (0,0) {$F$ or $Q$};
    \node (2) at (1,-1) {$T_t'$};
    \draw (1) -- (2); 
    \draw (1) -- (-0.5,-1);
    \node (3) at (-0.75,-1) {$\ldots$};
    \draw (1) -- (-1,-1);  

    \node (0) at (2,0) {$\longrightarrow$};

    \node[draw,rectangle] (4) at (4,0) {$F$};
    \node (5) at (5,-1) {$T_t'$};
    \draw (5) -- (4);
    \draw (4) -- (4,-1);
    \node (6) at (4.25,-1) {$\ldots$}; 
    \draw (4) -- (4.5,-1);
  \end{tikzpicture}

  \begin{tikzpicture}
    \node[draw,rectangle,minimum width=1.5cm] (1) at (0,0) {$F$ or $Q$};
    \node (2) at (0,-1) {$T_t'$};
    \draw (1) -- (2); 
    \draw (1) -- (0.5,-1);
    \draw (1) -- (-0.5,-1);  
    \node (0) at (2,0) {$\longrightarrow$};
    \node (4) at (4,0) {Impossible};
  \end{tikzpicture}
  \caption{The transformations for the \texttt{Force} procedures}
  \label{fig:transf}
\end{figure}
\end{proof}

\begin{corollary}
\label{cor:separator}
Let $\calH = (V,E)$ be a hypergraph and $T$ a consistent $PQF$-tree for $\calH$. Let $V' \subseteq V$ and $A = \{e \in E \mid e \cap V' \neq \emptyset \}$. Suppose that for all $e,f \in A$, $e \cap V' \subseteq f \cap V'$ or $f \cap V' \subseteq e \cap V'$. Then we can compute in polynomial time a $PQF$-tree $T'$ such that $\calF(T') = \{e_1 \ldots e_m \in \calF(T) \mid \forall i < j, e_i,e_j \in A \Rightarrow e_i \cap V' \subseteq e_j \cap V' \}$.
\end{corollary}
\begin{proof}
We want to compute $T'$ such that the frontiers of $T'$ are the frontiers of $T$ in which the edges of $A$ appear in increasing order with respect to inclusion relative to $V'$. The set $\{e \cap V' \mid e \in A\}$ is ordered by inclusion, thus it has a smallest element $V_1$ and a biggest element $V_2$. Obviously, $V_1 \subseteq V_2 \subseteq V'$. Furthermore, $A= \{e\in E\mid V_1\in e\}$ and $A_2 = \{e \in A \mid  V_2 \subseteq e \}$ is not empty. 

First use Lemma~\ref{lem:pqfsub} to find a $PQF$-subtree $S$ of $T$ whose leaves are exactly $\{ e \mid V_1 \subseteq e\} = A$. Then use Lemma~\ref{lem:pqfsub} again to find a $PQF$-subtree $R$ of $S$ whose leaves are exactly $A_2$. Now use the procedure $\texttt{Force}(S,R)$ to compute $S'$ as in Proposition~\ref{prop:force} and let $T'$ be the tree where we replace $S$ by $S'$ in $T$. As finding the right subtrees in $T$ can easily  be done in polynomial time by finding a least common ancestor and $\texttt{Force}$ is a polynomial time, it is clear that one can compute $T'$ as well. We now show that $T'$ has the desired properties.

To this end, let $\ell \in \calF(T')$. By definition of $T'$, we have that $\ell$ is also in $\calF(T)$. In addition, $\ell$ is of the form $\ell_1 \ell_A \ell_2$ with $\ell_A \in \calF(S')$. By definition of $\texttt{Force}$, $\ell_A$ is of the form $\sigma_1\sigma_2$ with $\sigma_2 \in \calF(R)$, that is, consisting only of the edges in $A_2$, which are maximal for the inclusion. Let $g \in A_2$. Let $e,f \in A$ with $e$ appearing before $f$ in $\ell_A$. If $e \cap V'$ is not included in $f \cap V'$, then there exists a $v' \in V_2$ such that $v' \in e$, $v' \notin f$ and $v' \in g$ since $V_2 \subseteq g$. That would lead to an inconsistent list which is a contradiction.

Reciprocally, let $\ell \in \{e_1 \ldots e_m \in \calF(T) \mid \forall i < j, e_i,e_j \in A \Rightarrow e_i \cap V' \subseteq e_j \cap V' \}$. $\ell$ is of the form $\ell_1 \ell_A \ell_2$ with $\ell_A \in \calF(S)$. As it is organized by inclusion relative to $V'$, the elements of $A_2$ should all lie at the end of $\ell_A$. Thus $\ell_A \in \{\sigma_1\sigma_2 \mid \sigma_2 \in \calF(R)\} = \calF(S')$ by Proposition~\ref{prop:force}. It follows that $\ell \in \calF(T')$.
\end{proof}

\subsection{Db-rootable hypergraphs}

In this section we will prove several structural properties of hypergraphs with disjoint branches decompositions which we will use in the algorithm in the next section.

\begin{definition}
Let $\calH=(V,E)$ be a hypergraph. For $e \in E$, we say that $\calH$ is {\em db-rootable} in $e$ if there exists a disjoint branches decomposition of $\calH$ rooted in $e$.
\end{definition}

The algorithm for the construction of disjoint branches decompositions will delete edges of hypergraphs. To this end we introduce the following notation.

\begin{definition}
For a hypergraph $\calH=(V,E)$ and an edge $e \in E$, we denote $H \setminus e$ the hypergraph $(V_e, E \setminus \{e\})$ where $V_e:=\bigcup_{e'\in E\setminus \{e\}} e'$.
For a set $A = \{e_1, \ldots, e_k\} \subseteq E$, we define $\calH \setminus A$ to be the hypergraph $((\calH \setminus e_1) \setminus \ldots )\setminus e_k$.
\end{definition}

We make the following observation which will simplify our arguments later.

\begin{observation}\label{obs:component}
 Let $\calH=(V,E)$ be a hypergraph and $e\in E$. Then $\calH$ is db-rootable in~$e$ if and only if for every connected component $C=(V_C,E_C)$ of $\calH\setminus e$ the hypergraph $C':=(V_C\cup e, E_C\cup\{e\})$ is db-rootable in $e$.
\end{observation}

The next lemma shows that deleting the root of a disjont branches decomposition decomposes a hypergraph along the components of the resulting join forest.

\begin{lemma}
\label{lem:cc}
Let $\calH$ be a hypergraph $\calH$ with a disjoint branches decomposition $\calT$ that is rooted in $e$. Let and $v_1,v_2$ be two vertices that appear in different trees $\calT_1$ and $\calT_2$ of the forest $\calT \setminus \{e\}$. Then $v_1$ and $v_2$ lie in different connected components of $\calH \setminus e$.
\end{lemma}
\begin{proof}
Suppose that there is a path $e_1, \ldots, e_k$ from $v_1$ to $v_2$ in $\calH \setminus e$. We show by induction on $i$, that $e_i$ is in $\calT_1$. It is true for $i=1$ because $v_1 \in e_1$, the disjointness of $\calT$ prevents $e_1$ from being in $\calT_2$. Now, assume that $e_i \in \calT_1$. As $e_i \cap e_{i+1} \neq \emptyset$, it means that $e_i$ and $e_{i+1}$ are in a common branch of $\calT$. If $e_{i}$ is an ancestor of $e_{i+1}$ in $\calT$, then by induction, $e_{i+1}$ is in $\calT_1$. If $e_{i+1}$ is an ancestor of $e_i$ in $\calT$, then either $e_{i+1}$ is in $\calT_1$ either $e_{i+1} = e$. However, $e_{i+1} \neq e$, because it is an edge of $\calH \setminus e$, so $e_{i+1} \in \calT_1$. 

Thus $e_k \in \calT_1$. However, it contradicts the disjointness of $\calT$ since $v_2 \in e_k$. Thus there is no path from $v_1$ to $v_2$ in $\calH \setminus e$: they are in two different connected components.
\end{proof}

Finally, we make one last observation on the roots of disjoint branches decompositions.

\begin{lemma}
\label{lem:cover}
If $\calH$ is db-rootable in $e$ and $\calH \setminus e = (V',E')$ has one connected component then there exists $e' \in E'$ such that $e \cap V' \subseteq e'$.
\end{lemma}
\begin{proof}
By Lemma~\ref{lem:cc}, in a disjoint branches decomposition of $\calH$ rooted in $e$, the edge $e$ has only one child $e'$. Thus, by connectedness, the vertices in $e$ that are not in $e'$ are only in $e$, thus they are not in $V'$ by definition. So $e \cap V' \subseteq e'$.
\end{proof}

\subsection{Separators}

By Observation~\ref{obs:component} we may deal with the components of $\calH\setminus e$ for a hypergraph $\calH$ and an edge $e$ independently. Thus we will in this section always assume that $\calH\setminus e$ just has a single component. We will consider restricted join paths that we call $A$-separators. In the following, all join paths will be denoted as ordered lists of edges, which corresponds to the notation in Section~\ref{subsec:pqf}.

\begin{definition}\label{def:separator}
Let $\calH=(V,E)$ be a hypergraph and let $\calP = a_1 \ldots a_m$ be a join path of $A \subseteq E$. We call $\calP$ an {\em $A$-separator} of $\calH$ if for all connected components $C = (V_C,E_C)$ of $\calH \setminus A$ we have that if $a_j \cap V_C \neq \emptyset$, then for all $i \leq j$, $a_i \cap V_C \subseteq a_j \cap V_C$.
\end{definition}

\begin{theorem}\label{thm:computeseparator}
There is a polynomial time algorithm \emph{$\texttt{ComputeSeparator}(\calH, A)$} that, given a hypergraph $\calH=(V,E)$ and a set $A \subseteq E$, computes an $A$-separator of $\calH$ if it exists and rejects otherwise.
\end{theorem}
\begin{proof}
We will iterate the algorithm described in Corollary~\ref{cor:separator}. We first compute a $PQ$-tree $T_0$ such that $\calF(T_0)$ is the set of all join paths for $A$ using Theorem~\ref{th:boothluecker}. Then, for each component $C=(V_C, E_C)$ of $\calH \setminus A$, we iteratively do the following: If there are edges $a_i$ and $a_j$ such that $a_i \cap V_C \nsubseteq a_j \cap V_C$ and $a_j \cap V_C \nsubseteq a_i \cap V_C$, then $\calH$ cannot have an $A$-separator and we reject. Otherwise, we can use the algorithm of Corollary~\ref{cor:separator} on the $PQF$-tree we have computed so far to construct a $PQF$-tree whose frontiers respect the order condition for the edges imposed by $C$ or rejects. 

An easy induction shows that if this algorithm does not reject at any point, then the computed $PQF$-tree $T$ has as frontiers all join paths that satisfy the order conditions of Definition~\ref{def:separator}. Thus we can choose one of these join paths arbitrarily as the desired $A$-separator.
\end{proof}

We will need $A$-separators with additional conditions.

\begin{definition}
Let $\calH=(V,E)$ be a hypergraph and let $A\subseteq E$. We call an $A$-separator $\calP = a_1 \ldots a_m$ of $\calH$ a {\em strong $A$-separator} if for all connected components $C = (V_C,E_C)$ of $\calH \setminus A$, $C' = (V_C, E_C \cup \{a_{l_C}\})$ is db-rootable in $a_{l_C}$ where $l_C = \max \{i \mid a_i \cap V_C \neq \emptyset \}$.
\end{definition}

As we will see, the existence of a strong $A$-separator for $\calH \setminus e$ and a certain set $A$ is equivalent to the property of being db-rootable in $e$. Thus we are interested in computing strong $A$-separators instead of arbitrary $A$-separators. Fortunately, it turns out that if there exists a strong $A$-separator, then all $A$-separators are strong and thus the algorithm of Theorem~\ref{thm:computeseparator} suffices.

\begin{proposition}
\label{prop:strongall}
Let $\calH=(V,E)$ be a hypergraph and $A\subseteq E$. If there exists a strong $A$-separator of $\calH$ then all $A$-separators of $\calH$ are strong.
\end{proposition}
\begin{proof}
Let $\calP$ be an $A$-separator and $C$ a connected component of $\calH\setminus A$. Let $e_C$ be the edge of the existing strong separator in which $C \cup \{e_C\}$ is db-rootable. Let $\calT_C$ be the corresponding disjoint branches decoposition rooted in $e_C$ and let $f_C$ be the last edge on $\calP$ such that $f_C \cap V_C \neq \emptyset$. As $e_C \cap V_C \neq \emptyset$, $e_C$ comes before $f_C$ on $\calP$. As $\calP$ is an $A$-separator, $e_C \cap V_C \subseteq f_C \cap V_C$. As only vertices of $V_C$ appear in the children of $e_C$ in $\calT_C$, we can replace the root $e_C$ of $\calT_C$ without breaking the connectedness condition of the vertices of $V_C$ or the disjointness of the branches (we do not change the branches of $\calT_C$). Doing this for all components of $\calH\setminus C$ shows that $\calP$ is strong.
\end{proof}

We will now show the main theorem that directly yields the algorithm for the computation of disjoint branches decomposition by reducing the construction of a disjoint branches decomposition to the construction of strong separators.

\begin{theorem}
\label{thm:correction}
Let $\calH = (V,E)$ be a hypergraph, $e \in E$, and $\calH' = (V',E') = \calH \setminus e$. Let furthermore $A_e:= \{e'\in E' \mid e\cap V' \subseteq e'\}$. Assume that $\calH'$ has only one connected component. Then $\calH$ is db-rootable in $e$ if and only if there exists a strong $A_e$-separator of $\calH'$.
\end{theorem}
\begin{proof}
Suppose $\calH$ is db-rootable in $e$. Let $\calT$ be a disjoint branches decomposition of $\calH$ rooted in $e$. By Lemma~\ref{lem:cover}, $A_e$ is not empty. Let $v \in e \cap V'$. We know that $v$ is contained in all edges $e'\in A_e$. Thus, by disjointness, the edges in $A_e$ are on the same branch of $\calT$. Moreover, we claim that $A_e$ is connected in $\calT$. To see this, suppose that $b \in E$ is between $a,c \in A_e$ on this branch. Then by connectedness, $e \cap V' \subseteq b$, so $b \in A_e$. Consequently, $A_e$ is connected and thus forms a path. Let $\calP = a_1 \ldots a_k$ be this path in $\calT$ in the direction from the root to the leaves of the tree. We claim that $\calP$ is a strong $A_e$-separator.

To this end, let $C = (V_C,E_C)$ be a connected component of $\calH' \setminus A_e$. We consider the forest obtained by removing $A_e$ in $\calT$. By Lemma~\ref{lem:cc}, vertices in different trees of this forest are in different connected component of $\calH' \setminus A_e$ as well. Thus there is a tree $\calT_C$ that contains all the edges in $E_C$. Let $a_{l_C}$ be the edge of $A_e$ to which the root of $\calT_C$ is connected. If $j > l_C$, then we claim that $a_j \cap V_C = \emptyset$. Assume this were not the case, then $\calT_C$ and the subtree of $a_{l_C}$ containing $a_j$ were not disjoint which is a contradiction to $\calT$ being a disjoint branches decomposition. Thus $a_j\cap V_C=\emptyset$. Now consider $i < j \leq l_C$ and let $v \in a_i \cap V_C$. We have $v \in a_{l_C} \cap V_C$ by the connectedness condition since $v$ appears in $\calT_C$ and, again by connectedness the connectedness condition, $v \in a_j \cap V_C$. Thus, $a_i \cap V_C \subseteq a_j \cap V_C$. It follows that $\calP$ is an $A_e$-separator. As $C \cup \{a_{l_C}\}$ is db-rootable in $a_{l_C}$ using the tree $\calT_C$, we have that $\calP$ is a strong $A_e$-separator.

Assume now that there is a strong $A_e$-separator $\calP = a_1 \ldots a_k$. For a connected component $C$ of $\calH' \setminus A_e$, let $l_C := \max \{ i \mid a_i \cap V_C \neq \emptyset\}$. By definition, there exists a disjoint branches decomposition $\calT_C$ for $C \cup \{a_{l_C}\}$ rooted in $a_{l_C}$. We construct a disjoint branches decomposition for $\calH$ as follows:
\begin{itemize}
\item We root the path $e a_1 \ldots a_k$ in $e$.
\item For each connected component $C$ of $\calH' \setminus A_e$, we connect the root of $\calT_C$ to $a_{l_C}$.
\end{itemize}

We claim that the resulting tree $\calT$ is a disjoint branches decomposition. We first show that the branches of $\calT$ are disjoint. Indeed, if $a,b$ are edges in two different branches then two cases can occur: Either $a$ and $b$ are in two different connected component of $\calH' \setminus A_e$. But then they are disjoint, because they lie in different components of $\calH$ by construction. Otherwise, let $a$ be in a connected component $C$ of $\calH' \setminus A_e$ and let $b$ be in $A_e$. But as $b$ is on a different branch as $a$, it follows that $b$ comes after $a_{l_C}$ on $\calP$. Thus $b \cap V_C = \emptyset$ and it follows that $a \cap b = \emptyset$ since $a \subseteq V_C$. Thus the branches of $\calT$ are disjoint.

Now we show that $\calT$ is a join tree, i.e., it satisfies  the connectedness condition for all vertices. For a connected component $C$ of $\calH' \setminus A_e$, for a vertex $v \in V_C$, its connectedness is ensured along $\calP$ since $\calP$ is a join path and in $\calT_C$ since $\calT_C$ is a join tree. Furthermore, by construction $v\in e_{l_dC}$ so the edges containing $v$ are connected in $\calT$. For a vertex $v$, which does not appear in any $V_C$, that is, which only appears in $A_e$, its connectedness is ensured by the fact that $\calP$ is a join path for $A_e$.
\end{proof}

We now turn the proof of Theorem~\ref{thm:correction} into the algorithm for the computation of disjoint branches decompositions.

\begin{corollary}
  \label{cor:computedb}
  There is a polynomial time algorithm \texttt{\em ComputeDB}$(\calH,e)$ that, given a hypergraph $\calH = (V,E)$ and and edge $e \in E$, returns a disjoint branches decomposition of $\calH$ rooted in $e$ if it exists and rejects otherwise.
\end{corollary}
\begin{proof}
We give a pseudo-code of \texttt{ComputeDB} in Algorithm~\ref{alg:computedb}. Following Observation~\ref{obs:component}, we deal with the connected component of $\calH \setminus e$ independently. For each connected component of $\calH$, we compute $A_e$. We reject if it is empty since it means by Lemma~\ref{lem:cover} that $\calH$ is not db-rootable in $e$. Otherwise we compute an $A_e$-separator $\calP = a_1 \ldots a_k$ of $\calH$ in polynomial time with the procedure \texttt{ComputeSeparator} of Theorem~\ref{thm:computeseparator} and recursively check its strongness by trying to compute a disjoint branches decomposition $\calT_C$ of $C \cup \{a_{l_C}\}$ rooted in $a_{l_C}$ where $l_C = \max \{i \mid a_i \cap V_C \neq \emptyset \}$. Following the proof of Theorem~\ref{thm:correction}, we construct a disjoint branches decomposition of $\calH$ rooted in $e$ by rooting the path $e a_1 \ldots a_k$ in $e$ and connecting the root of $\calT_C$ to $a_{l_C}$.

Theorem~\ref{thm:correction} ensures both that if the algorithm does not reject then the computed decomposition is a disjoint branches decomposition of $\calH$ rooted in $e$ and that if such a decomposition exists then the algorithm won't reject. 

\begin{algorithm}
\caption{The algorithm \texttt{ComputeDB}}
\begin{algorithmic}
  \State {\texttt{ComputeDB}}($\calH = (V,E)$, $e$) =
  \If {$|E| = 1$} 
  \Return the tree with the only vertex $e$
  \Else
  \For{each connected component $\calH_i = (V_i,E_i)$ of $\calH \setminus e$}
  \State{$A_e \leftarrow \{ e' \in E_i \mid e \cap V_i \subseteq e' \}$}
  \If {$A_e = \emptyset$} \State{Fail.} \EndIf
  \State{$\calP \leftarrow$ \texttt{ComputeSeparator}$(\calH_i,A_e)$}
  \State{$\calT_i \leftarrow \calP$}
  \For{each connected component $C = (V_C,E_C)$ of $\calH_i \setminus A_e$}
  \State{$l_C \leftarrow \max \{j \mid a_j \cap V_C \neq \emptyset \}$ (where $\calP = a_1 \ldots a_k$})
  \State{$C' \leftarrow (V_C \cup a_{l_C}, E_C \cup \{a_{l_C}\})$}
  \State{$\calT_C \leftarrow \texttt{ComputeDB}(C',a_{l_C})$}
  \State{connect $\calT_C$ to $\calT_i$ in $a_{l_C}$}
  \EndFor
  \EndFor

  \Return the tree rooted in $e$ having $\calT_1,\ldots,\calT_p$ as children
  \EndIf
\end{algorithmic}
\label{alg:computedb}
\caption{The algorithm \texttt{ComputeDB} of Corollary~\ref{cor:computedb}.}
\end{algorithm}                 

The algorithm runs in polynomial time since there is at most $|E|$ recursive calls of \texttt{ComputeDB} and that each call is in polymomial time. 
\end{proof}

\section{Some negative results on generalizations}\label{sct:extensions}

In this section we will discuss several approaches to generalizing the counting algorithm for hypergraphs with disjoint branches to more general classes of hypergraphs. Unfortunately, all these results will be negative as we will show hardness results for all extensions we consider. We still feel these results are worthwhile because they might help in guiding future research to classes of hypergraphs that are better suited for $\sSAT$.

Our main technical tool will be a helpful result by Samer and Szeider \cite{Samer:2010wc}. For a hypergraph $\calH=(V,E)$, let $\calH^*$ be the hypergraph $\calH^*:=(V\cup \{x\} , E\cup \{V\cup \{x\})$ where $x$ is a new vertex. We use the following result of Samer and Szeider.

\begin{lemma}[\cite{Samer:2010wc}]\label{lem:SamerS10}
 Let $F$ be a CNF-formula with hypergraph $\calH$. Then we can in polynomial time construct a CNF-formula $F'$ with hypergraph $\calH^*$ such that 
 \begin{itemize}
  \item $F$ has a satisfying assignment if and only if $F'$ has one, and
  \item the number of satisfying assignments of $F$ can be computed from those of $F'$ in polynomial time.
 \end{itemize}
\end{lemma}

Note that Lemma \ref{lem:SamerS10} is not proved explicitly in \cite{Samer:2010wc} but follows from the proof of Proposition 1 in that paper.

\subsection{Undirected paths}

In this section we consider a generalization of disjoint branches in the following way: We call a hypergraph $\calH$ \emph{undirected path acyclic} if there is a join tree $(\calT, \lambda)$ of $\calH$ such that for every $v\in V$ the edge set $\{e\in E\mid v\in e\}$ forms an undirected path in $\calT$. Undirected path acyclicity is a seemingly natural generalization of disjoint branches acyclicity by allowing undirected paths instead of directed paths. Unfortunately, as we will see this generalization makes our counting problem hard and in fact even $\SAT$ hard. To show this we will use the following notions: We call a CNF-formula \emph{read-twice} if every variable appears at most twice in it. 

\begin{lemma}[\cite{HuntS90}]\label{lem:hard}
 $\SAT$ for read-twice-formulas is $\NP$-hard. 
\end{lemma}

We show that undirected path acyclicity does not not allow tractable $\SAT$ and thus it is not a good generalization in our setting.

\begin{theorem}\label{thm:hard}
 $\SAT$ on undirected path acyclic CNF-formulas is $\NP$-complete. 
\end{theorem}
\begin{proof}
 Let $F$ be a read-twice-formula. We construct $F'$ as in Lemma \ref{lem:SamerS10} and claim that $F'$ is undirected path acyclic. The corresponding join tree has $V\cup \{x\}$ as its root and all other edges are leaves of the join tree that are connected to the root. Then every variable appears in at most $2$ leaves because $F$ is read-twice. Since all variables also appear in the root, every variable induces a path in this join tree. So $F'$ is undirected path acyclic. 
 
 The claim is now easy to see: $\SAT$ for $F$ and $F'$ is equivalent, but with Lemma \ref{lem:hard} $\SAT$ for $F$ is hard. This gives the desired reduction for $\NP$-hardness.
\end{proof}

\subsection{Allowing limited intersections and appearance in several branches}

A natural way of generalizing disjoint branches decompositions is allowing limited intersections between branches. We will show that this approach leads to hard counting problems even if we only allow the intersection to contain one variable.

\begin{lemma}\label{lem:boundedintersections}
 $\sSAT$ is $\sP$-hard for CNF-formulas that have a join tree in which the branches may may have a pairwise intersection containing on variable.
\end{lemma}
\begin{proof}
 We reduce from $\mathrm{\#VertexCover}$ which can alternatively be interpreted as $\stwoSAT$ on monotone formulas and is well known to be $\sP$-hard. So let $F$ be a monotone $2$-CNF formula. We construct $F'$ as in Lemma \ref{lem:SamerS10}. Let $(\calT, \lambda)$ be a join tree of $F'$ in which the edge $\var{F'}$ is the root and all other edges are leaves. Since $F$ is monotone and we may assume that no clause appears twice in it, the leaves intersect in at most one variable which completes the proof.
\end{proof}

\section{Conclusion}

We have presented a new structural class of tractable $\sSAT$ instances, those whose hypergraphs admit a disjoint branches decomposition. To this end, we also invested a considerable amount of work into an algorithm that computes the decompositions. 

Several questions remain, the most obvious open problem certainly being the complexity of $\sSAT$ on $\beta$-acyclic hypergraphs. Can one show a $\sP$-completeness result or a polynomial time algorithm for this case?

Another aim for future work is trying to turn the disjoint branches property into a hypergraph width measure such that $\sSAT$---or even $\SAT$---for the hypergraphs for which this width measure is bounded is tractable? Can we construct this measure to even allow fixed-parameter tractability? Note that it is known that the parameterization by incidence cliquewidth both does not allow fixed-parameter tractability~\cite{OrdyniakPS13}.

More generally, we feel that it is very desirable to understand the tractability frontier for $\SAT$ and $\sSAT$ with respect to structural restrictions better overall. Is there a width measure that generalizes both hypergraphs with disjoint branches and incidence cliquewidth that leads to tractable $\sSAT$? A natural candidate would be $\beta$-hypertree width (see Figure~\ref{FIG:hierarchy}). Are there other classes of hypergraphs incomparable to those studied so far that give large structural classes of tractable $\sSAT$-instances? Note that there is a similar line in the area of constraint satisfation (see e.g.~\cite{miklos08} for an overview) that has been very successful but unfortunately does not apply directly.

\bibliographystyle{alpha}
\bibliography{bibfile}

\begin{thebibliography}{FMR08}

\bibitem[BB12]{BraultBaron12}
J.~Brault-Baron.
\newblock {A Negative Conjunctive Query is Easy if and only if it is
  Beta-Acyclic}.
\newblock In {\em {Computer Science Logic, 26th International Workshop/21st
  Annual Conference of the EACSL}}, pages 137--151, 2012.

\bibitem[BL76]{boothlueker}
K.S. Booth and G.S. Lueker.
\newblock Testing for the consecutive ones property, interval graphs, and graph
  planarity using {PQ-tree} algorithms.
\newblock {\em Journal of Computer and System Sciences}, 13(3):335--379, 1976.

\bibitem[BL03]{brandloz}
A~Brandst\"{a}dt and Vadim~V Lozin.
\newblock On the linear structure and clique-width of bipartite permutation
  graphs.
\newblock {\em Ars Combinatoria}, 67(1):273--281, 2003.

\bibitem[CGH09]{CohenGH2009}
D.A. Cohen, M.J. Green, and C.~Houghton.
\newblock Constraint representations and structural tractability.
\newblock In {\em Proceedings of the 15th international conference on
  Principles and practice of constraint programming}, CP'09, pages 289--303,
  2009.

\bibitem[Dur12]{Duris12}
David Duris.
\newblock Some characterizations of $\gamma$ and $\beta$-acyclicity of
  hypergraphs.
\newblock {\em Inf. Process. Lett.}, 112(16):617--620, 2012.

\bibitem[Fag83]{Fagin-83}
R.~Fagin.
\newblock Degrees of acyclicity for hypergraphs and relational database
  schemes.
\newblock {\em Journal of the ACM}, 30(3):514--550, 1983.

\bibitem[FG06]{FlumG06}
J.~Flum and M.~Grohe.
\newblock {\em {Parameterized Complexity Theory}}.
\newblock Springer-Verlag New York Inc, 2006.

\bibitem[FMR08]{FMR-08}
E.~Fischer, J.A. Makowsky, and E.V. Ravve.
\newblock Counting truth assignments of formulas of bounded tree-width or
  clique-width.
\newblock {\em Discrete Applied Mathematics}, 156(4):511--529, 2008.

\bibitem[GP04]{GottlobP01}
G.~Gottlob and R.~Pichler.
\newblock {Hypergraphs in Model Checking: Acyclicity and Hypertree-Width versus
  Clique-Width}.
\newblock {\em SIAM Journal on Computing}, 33(2), 2004.

\bibitem[IS90]{HuntS90}
H.B.~Hunt III and R.E. Stearns.
\newblock {The Complexity of Very Simple Boolean Formulas with Applications}.
\newblock {\em SIAM Journal on Computing}, 19(1):44--70, 1990.

\bibitem[Mik08]{miklos08}
Z.~Mikl{\'o}s.
\newblock {\em {Understanding Tractable Decompositions for Constraint
  Satisfaction}}.
\newblock PhD thesis, University of Oxford, 2008.

\bibitem[OPS13]{OrdyniakPS13}
S.~Ordyniak, D.~Paulusma, and S.~Szeider.
\newblock {Satisfiability of acyclic and almost acyclic CNF formulas}.
\newblock {\em Theoretical Computer Science}, 481:85--99, 2013.

\bibitem[PSS13]{PaulusmaSS13}
D.~Paulusma, F.~Slivovsky, and S.~Szeider.
\newblock {Model Counting for CNF Formulas of Bounded Modular Treewidth}.
\newblock In {\em 30th International Symposium on Theoretical Aspects of
  Computer Science, STACS 2013}, pages 55--66, 2013.

\bibitem[Rot96]{Roth96}
D.~Roth.
\newblock On the hardness of approximate reasoning.
\newblock {\em Artificial Intelligence}, 82(1–2):273 -- 302, 1996.

\bibitem[Sli14]{Slivovsky14}
F.~Slivovsky.
\newblock personal communication, 2014.

\bibitem[SS10]{Samer:2010wc}
M.~Samer and S.~Szeider.
\newblock {Algorithms for propositional model counting}.
\newblock {\em Journal of Discrete Algorithms}, 8(1):50--64, 2010.

\bibitem[SS14]{SlivovskyS13}
F.~Slivovsky and S.~Szeider.
\newblock {Model Counting for Formulas of Bounded Clique-Width}.
\newblock to appear, 2014.

\end{thebibliography}

\end{document}